\documentclass[submission,copyright,creativecommons]{eptcs}
\providecommand{\event}{GandALF 2017} 
\usepackage{breakurl}             
\usepackage{underscore}           
\usepackage{amsmath}
\usepackage{amssymb}
\usepackage{amsthm}
\theoremstyle{plain}
\newtheorem{thm}{Theorem}
 \theoremstyle{definition}
  \newtheorem{example}[thm]{Example}
  \theoremstyle{definition}
  \newtheorem{defn}[thm]{Definition}
  \theoremstyle{plain}
  \newtheorem{lemma}[thm]{Lemma}
  \theoremstyle{remark}
  \newtheorem{remark}[thm]{Remark}
    \theoremstyle{remark}
  
  \theoremstyle{plain}
  
  \theoremstyle{plain}
  \newtheorem{conjecture}[thm]{Conjecture}

\title{The Descriptive Complexity \\ of Modal $\mu$ Model-checking Games}
\author{Karoliina Lehtinen
\institute{University of Kiel\\ Kiel, Germany}
\email{kleh@informatik.uni-kiel.de}
}
\def\titlerunning{The Descriptive Complexity of Modal $\mu$ Model-checking Games}
\def\authorrunning{Karoliina Lehtinen}
\begin{document}
\maketitle

\begin{abstract}
This paper revisits the well-established relationship between the modal $\mu$ calculus $L_\mu$ and parity games to show that it is even more robust than previously known. It addresses the question of whether the descriptive complexity of $L_\mu$ model-checking games, previously known to depend on the syntactic complexity of a formula, depends in fact on its semantic complexity. It shows that up to formulas of semantic complexity $\Sigma^{\mu}_2$, the descriptive complexity of their model-checking games coincides exactly with their semantic complexity. Beyond $\Sigma^{\mu}_2$, the descriptive complexity of the model-checking parity games of a formula $\Psi$ is shown to be an upper bound on the semantic complexity of $\Psi$; whether it is also a lower bound remains an open question.

 \end{abstract}

\section{Introduction}

The modal $\mu$-calculus~\cite{Kozen1983333}, written $L_\mu$, is a verification logic consisting of a simple modal logic augmented with a least fixpoint $\mu$ and its dual $\nu$. It is a prime example of the intersection between logic and games: In its model-checking games, the antagonism between the verifying player Even and her opponent Odd describes the duality between conjunctions and disjunctions, and between least and greatest fixpoints. The complexity of the winning condition -- a parity condition over a set of priorities -- corresponds to the complexity of the formula, as measured by its index, \emph{i.e.} the number of alternations between least and greatest fixpoints. 

 The index of a formula is a robust measure of a formula's complexity: for each fixed index, the model-checking problem is in {\sc P}, but no fixed index suffices to express all $L_\mu$-expressible properties~\cite{bradfield-strict}. The decidability of the \emph{semantic index} of a formula, that is to say the least index of any equivalent formula, is a long-standing open problem. Despite the efforts put towards solving it and its automata-theoretic counterpart, the Mostowski--Rabin index problem for parity automata, only the low levels of the alternation hierarchy are currently known to be decidable.
Formulas semantically in $\mathit{ML}$, the fragment without fixpoints, and those in $\Pi^{\mu}_1$ and $\Sigma^{\mu}_1$, the fragments with only one type of fixpoint, were first shown to be decidable by Otto~\cite{otto1999eliminating} and K\"usters and Wilke~\cite{kusters2002deciding} respectively. Alternative formula-focused decidability proofs for both were derived in~\cite{lehtinen2015deciding}.
Some results in automata theory~\cite{colcombet2013deciding} can be
interpreted within the $L_\mu$ context to yield another decidability result: given a formula in the $\Pi^{\mu}_2$ alternation class,
it is decidable whether it is equivalent on ranked trees to a formula in the $\Sigma^{\mu}_2$ alternation class. An alternative proof with a topological account was given in~\cite{skrzypczakdeciding}, while in~\cite{lehtinen2017accepted} a game-theoretic characterisation of $\Sigma^{\mu}_2$ extends the same result onto unranked trees: $\Sigma^{\mu}_2$ is decidable for $\Pi^{\mu}_2$.
 
 This paper revisits one of the fundamental concepts of $L_\mu$ literature in light of the index-problem: the relationship between $L_\mu$ formulas and parity games. Parity games can be used to describe the semantics of $L_\mu$ and are therefore a recurring topic in $L_\mu$ literature. This relationship is complexity-preserving, in the following sense: On the one hand, the model-checking games for
$L_\mu$ formulas of index $I$ are parity games with priorities from $I$; on the other, a $L_\mu$ formula of
index $I$ suffices to describe the winning regions of such games. As a result, the descriptive complexity of the
model-checking games of a formula $\Psi$ is bounded by the index of $\Psi$.
This paper asks whether this complexity-preserving relationship can be extended to also account for the \emph{semantic} complexity of a formula: can the bounds on the complexity of model-checking games be tightened to the semantic index of $\Psi$?

In its first part (Section \ref{first-part}), this paper argues that  if
the winning regions of the model-checking games for $\Psi$ are described by $\Phi$ -- written
$\Phi$ interprets $\Psi$ -- then $\Psi$ is semantically as simple as $\Phi$, \emph{i.e.} equivalent to a formula of the same index as $\Phi$. In other words, the descriptive complexity of the model-checking games of a formula is an upper bound on the semantic complexity of the formula. 

The second part of the paper (Section \ref{second-part}) considers the converse: if a formula is semantically of index $I$, then is it interpreted by a formula of index $I$? That is to say, do semantically simple formulas generate equally simple model-checking games? This is shown to be the case for the first semantic levels of the alternation hierarchy. In some cases, the input formula must be transformed into disjunctive form~\cite{janin1995automata} first, a normal form for $L_\mu$:
\begin{itemize}
\item If a $L_\mu$ formula is equivalent to a modal formula, it is interpreted by a modal formula;
\item If a disjunctive $L_\mu$ formula is equivalent to a $\Pi^\mu_1$ formula, it is interpreted by a $\Pi^\mu_1$ formula;
\item If a disjunctive $L_\mu$ formula is equivalent to a $\Sigma^\mu_2$ formula, it is interpreted by a $\Sigma^\mu_2$ formula;
\end{itemize}

Interestingly, these theorems are based on characterising the formulas semantically in $\mathit{ML},\Pi^{\mu}_1$ and $\Sigma^{\mu}_2$ as those for which the model-checking games coincide with variations of parity games that characterise the class. 
 For these (semantic) alternation-classes at least, the descriptive complexity of the model-checking games of a formula $\Psi$ corresponds exactly to the semantic index of $\Psi$, rather than its syntactic index.

Section \ref{sec-general-interpretation} considers the general case and shows that if a \emph{co-disjunctive} formula is equivalent to a \emph{disjunctive} formula of index $I$, then it is also interpreted by a disjunctive formula of index $I$. I conjecture that the same holds for disjunctive input formulas, but the currently available techniques do not suffice to prove this.

In short, this paper shows that the relation between $L_\mu$ and parity games can be extended in a natural way to account for the $L_\mu$ index problem: the descriptive complexity of the model-checking games of $\Psi$ is an upper bound on the semantic complexity of $\Psi$. For several non-trivial cases, and in particular beyond the alternation classes currently known to be decidable, the descriptive complexity of the model-checking games generated by $\Psi$ coincides exactly with the semantic complexity of $\Psi$. In the general case, the index of an equivalent disjunctive formula bounds the descriptive complexity of a co-disjunctive formula. Overall, the $L_\mu$--parity game relationship is even more robust than previously thought.

One practical consequence is a method for simplifying formulas by analysing their model-checking games (see Example \ref{ex}). More generally, if one can show that a fragment $F\subset L_\mu$ generates model-checking games of bounded descriptive complexity, then $F$ is embedded in a fixed level of the $L_\mu$ alternation hierarchy.  Furthermore, the study of the descriptive complexity of the model-checking games  is a novel approach to the $L_\mu$-index problem. Although the three decision procedures for $\mathit{ML}$, $\Pi^{\mu}_1$ and $\Sigma^{\mu}_2$ (for $\Pi^{\mu}_2$ formulas) each uses a completely different strategy, describing the winning regions of formulas semantically in the target class subsumes all three. Extending the interpretation theorems of this paper seems likely to lead to characterisations of higher $L_\mu$ alternation classes, providing a stepping stone towards further decidability results. In particular, a generalisation of the $\Sigma^{\mu}_2$ interpretation theorem would yield a reduction of the decidability of the index problem to a boundedness question, in the style of what Colcombet \emph{et. al.} achieved for non-deterministic automata~\cite{colcombet2008non}.

The descriptive complexity of parity games has already been considered in~\cite{dawar2008} where the authors ask which formalisms can describe the winning regions of classes of parity games without a fixed index. They show that guarded second order logic suffices in the general case while least fixed point logic does not. On finite game arenas, the winning regions are definable in least fixed point logic if and only if solving parity games is in {\sc P}. In contrast, the winning regions of the classes of parity games considered here are trivially $L_\mu$ expressible, and the focus is on exactly which fragment of $L_\mu$ is necessary.\\

The following section fixes the notation and terminology used throughout. Some familiarity with $L_\mu$ and parity games is assumed -- see for example~\cite{bradfield2007modal} for an introduction.

\section{Preliminairies}

\subsection{$L_\mu$ syntax}

Fix countably infinite sets $\mathit{Prop}=\{P,Q,...\}$ of propositional variables, and $\mathit{Var}=\{X,Y,...\}$ of fixpoint variables. For the sake of clarity and conciseness the scope of this paper is restricted to the unimodal $L_\mu$; however, the results presented here extend easily to the multi-modal case.

 \begin{defn}\emph{($L_\mu$)}
The syntax of unimodal $L_{\mu}$ is given by:
\[ \phi:= 
P \ \ | \ \ 
X \ \ | \ \ 
\neg P \ \ | \ \ 
\phi\wedge\phi \ \ | \ \ 
\phi\vee\phi \ \ | \ \ 
\Diamond \phi \ \ | \ \ 
\Box \phi \ \ | \ \ 
\mu X.\phi \ \ | \ \ 
\nu X.\phi \ \ | \ \ 
\bot \ \ | \ \ 
\top
\]
 \end{defn}
 Conjunctions take precedence over disjunctions. The scope of fixpoint bindings extends as far as possible to the right while the scope of modalities
extends as little as possible to the right. For example, $\mu X.  \Diamond X \wedge C \vee B $ is parsed as $\mu X. (((\Diamond X)\wedge C) \vee B)$. Note that here negation can only be applied to propositional
variables.



A fixpoint variable $X$ which does not appear in the scope of $\mu X.\phi$ or $\nu X.\phi$ is
said to be \emph{free}, and the set of free fixpoint variables of a formula $\psi$ is written $\mathit{Free}_\psi$. A sentence is a formula without free fixpoint variables.
A fixpoint $\mu$ or $\nu$ \emph{binds} a variable $X\in \mathit{Free}_\phi$ in $\mu X.\phi$ or $\nu X.\phi$ respectively.
The formula $\phi$ is called the binding formula of $X$, and will be written $\phi_X$.
For notational purposes, the binding formula of $X$ will be treated as an immediate subformula of $X$. The \emph{parse tree} of a sentence is the tree with the subformulas of $\Psi$ for nodes, rooted at the formula itself, and where the children of a node
consist of the parse-trees of its immediate subformulas. In other words it consists of the formula, written as a tree, with back-edges from fixpoint variables to their bindings.

A formula is \emph{guarded} if every fixpoint variable is in the scope of a modality within its
binding. Without loss of expressivity~\cite{MateescuRadu2002,kupferman2000automata}, all $L_\mu$ formulas are assumed to be in guarded form.\\

Disjunctive form is a normal form for $L_\mu$ formulas~\cite{janin1995automata} which imposes some additional structure onto the $L_\mu$ syntax. The dual definition yields co-disjunctive formulas. Every $L_\mu$ can be turned into a disjunctive formula, although the transformation may not preserve complexity~\cite{lehtinendisjunctive}. Unlike in general $L_\mu$ model-checking games, in the model-checking games of disjunctive formulas, the verifying player has strategies which only agree with one play per branch~\cite{lehtinen2015deciding}.

\subsection{Alternations}

The definition of alternation depth, here referred to as index to match the automata-theoretic terminology,
is meant to only capture alternations which generate algorithmic complexity, matching the definition given in Niwi\'nski 1986~\cite{niwinski1986fixed}.
The presentation here emphasizes the relationship of a formula's index to the priorities in parity games.
A thorough discussion on how to define alternation depth, and comparison of definitions used in the literature
can be found in Bradfield and Stirling 2007~\cite{bradfield2007modal}.

\begin{defn}\emph{(Priority assignment and index)}
A priority assignment for a sentence $\Psi$ is a mapping $\Omega: \mathit{Var}_\Psi \rightarrow \{m,...,q-1,q\}$, where $m\in \{0,1\}$ and $q$ is a positive integer, of priorities to
the fixpoint variables $\mathit{Var}_\Psi$ of $\Psi$ such that $\mu$-bound variables receive odd priorities while $\nu$-bound variables
receive even priorities.
A priority assignment is order preserving if whenever $X$ is free in the formula binding $Y$, $\Omega(X)\geq \Omega(Y)$ holds.

A formula has index $I$ if it has an order preserving priority assignment with co-domain $I$.
 A formula has semantic index $I$ if it is equivalent to a formula of index $I$.
\end{defn}

\begin{defn}\emph{(Alternation hierarchy)}
 The base of the alternation hierarchy is $\mathit{ML}$, the modal fragment of $L_\mu$, consisting of formulas
 without any fixpoints.
 Formulas with only $\nu$-bound fixpoints, or only $\mu$-bound fixpoints respectively, have index $\{0\}$, or $\{1\}$,
 corresponding to the alternation classes $\Pi^{\mu}_1$, or $\Sigma^{\mu}_1$. 
 The classes $\Pi^{\mu}_i$ and $\Sigma^{\mu}_i$ for positive even
 $i$ correspond to formulas with indices $\{1,...,i\}$ and $\{0,...,i-1\}$ respectively, while for odd $i$ they correspond to
 formulas with indices $\{0,...,i-1\}$ and $\{1,...,i\}$ respectively.

\end{defn}

\begin{example}
The formula $\mu X. \nu Y. (\Box Y \wedge \mu Z. \Box (X \vee Z))$ accepts the order-preserving assignment $\Omega(X)=1,\Omega(Y)=0$ and $\Omega(Z)=1$, and has therefore index $\{0,1\}$ and is in $\Sigma^{\mu}_2$. It is however equivalent to $\mu X. \Box X$ (see Section \ref{semantics}) and is therefore semantically in $\Sigma^{\mu}_1$.
\end{example}

\begin{thm}\cite{bradfield-strict,lenzi-strict}
 The alternation hierarchy is strict: for each index $I=\{m,...,q\}$, where $m\in \{0,1\}$, there are formulas that
 are not equivalent to a formula with smaller index.
\end{thm}

Deciding an index $I$ means deciding, given an arbitrary $L_\mu$ formula, whether it is equivalent to any formula of index $I$.

\subsection{$L_\mu$ semantics and parity games}\label{semantics}

We define the semantics of $L_\mu$ in terms of its model-checking games, known as parity games. The correspondence with the classical semantics is standard~\cite{bradfield2007modal}.

$L_\mu$ formulas operate on regular unranked trees.For unimodal $L_\mu$, these trees do not have edge-labels.

\begin{defn}\emph{(Trees)}
 An unranked tree $\mathfrak{T}=(S,E,r,P)$, rooted at $r\in S$ consists of:
 a set of states $S$,
 a successor relation $E\subseteq S\times S$, and
 a labelling $P:S \rightarrow \mathcal{P}(\mathit{Prop}_{\mathfrak{T}})$ mapping states to subsets of a finite set $\mathit{Prop}_{\mathfrak{T}}\subset \mathit{Prop}$.
For each state $s\in S$, the set  $\{w\in S\mid \exists  w_1,\ldots,w_k.\;(w,w_1)\in E,(w_1,w_2)\in E,\ldots (w_k,s)\in E\}$ of its ancestors, is finite and
well-ordered with respect to the transitive closure of $E$.
\end{defn}
The scope of this paper is restricted to regular trees with finite but unbounded branching, which can be finitely represented as trees with back edges. Since $L_\mu$ enjoys a finite model property~\cite{kozen1988finite}, the index problem on trees with infinite branching reduces to the index problem on trees with finite branching.

\begin{defn}\emph{(Parity game)}
A parity game arena is $A=(V,E,v_\iota,\Omega)$ consisting of: a set $V$ of states partitioned into those belonging to Even, $V_e$ and those belonging to Odd, $V_o$;
an edge relation $E\subseteq V\times V$; an initial node $v_\iota$; and a priority assignment $\Omega: V \rightarrow I$.
The co-domain $I$ of $\Omega$ is said to be the index of $A$. 
$I$ is a prefix of the natural numbers, starting at either $0$ or $1$.

A play in a parity game is a potentially infinite sequence of successive positions starting with $v_\iota$. A play is finite if it ends in a position without successors; then the owner of the position loses. For infinite plays, the winner depends on the highest priority seen infinitely often, called the dominant priority. The winner is the player of the
parity of the dominant priority.

A positional (or memoryless) strategy $\sigma$ for Even (and similarly for Odd) consists of a choice $\sigma(v)$ of successor at the nodes $v$ in $V_e$ (or $V_o$).
A play $\pi=v_0,v_1,...$, where $v_0=v_\iota$, agrees with a strategy $\sigma$
if, at every position $v_i\in V_e$ (or $V_o$) along the play, $v_{i+1}=\sigma(v_i)$.

A pair of strategies, one for each player, induces a unique play.
A strategy for a player is winning if every play that agrees with it is winning for this player.
Parity games are positionally determined~\cite{emerson1991tree,martin1975borel}: positional strategies suffice and exactly one of the players has a winning strategy
from every position.
A parity game $G(A)$ on arena $A$ is said to be winning for the player with a winnning strategy.
\end{defn}

We define the semantics of $L_\mu$ in terms of winning strategies in parity games:
For an unranked tree $\mathfrak{T}$ and formula $\Psi$, we say that $\mathfrak{T}$ satisfies $\Psi$, written $\mathfrak{T}\models \Psi$ if and only if Even has a winning strategy in the model-checking parity game $G(\mathfrak{T}\times \Psi)$ defined below. 

\begin{defn}\emph{(Model-checking parity game)}
 Let $\Psi$ be a sentence of $L_\mu$ and $\Omega_{\Psi}$ a priority assignment with co-domain $I$ on the fixpoint variables of $\Psi$.
 Then, for any unranked tree $\mathfrak{T}$, define the model-checking parity game arena $\mathfrak{T}\times \Psi$ as follows:
 $\mathfrak{T}\times \Psi = (V,E,v_\iota,\Omega)$ where:
 \begin{itemize}
  \item $V$ is the set of states $(s, \phi)$, where $s$ is a state of $\mathfrak{T}$ and $\phi$ is a subformula of $\Psi$.
  \item $V_o$ consists of positions $(s, \phi \wedge \psi)$ and $(s, \Box \phi)$ while $V_e=V\setminus V_o$;
  \item Positions $(s, C)$, for $C$ a literal, are terminal and if $C\in P(s)$, belong to Odd; else it belongs to Even;                                                                                                                                                                                                                                                                                                                                                                                                                                                                                                                                                                                                       
  \item There is an edge from $(s, \phi \vee \psi)$ to $(s, \phi)$ and $(s,\psi)$;
  
  an edge from $(s,\phi \wedge \psi)$ to $(s, \phi)$ and $(s,\psi)$;
  
  an edge from $(s, X)$, $(s, \mu X.\phi_X)$ and $(s, \nu X.\phi_X)$ to $(s, \phi_X)$;
  
  an edge from $(s, \Diamond \phi)$ and $(s, \Box \phi)$ to $(s' , \phi)$ for every successor $s'$ of $s$;

  \item $v_\iota$ the initial position is $(r, \Psi)$, where $r$ is the root of $\mathfrak{T}$;
  \item $\Omega$ assigns $\Omega_\Psi(X)$ to positions $s\times X$; $\Omega$ assigns the minimal priority of $I$ to all other positions.
 \end{itemize}

\end{defn}

\begin{defn}\emph{(Semantics of $L_\mu$ via games)}
If $\Omega_{\Psi}$ is an order-preserving priority assignment for $\Psi$, then for all $\mathfrak{T}$, Even has a winning strategy in the parity game $G(\mathfrak{T}\times \Psi)$ if and only if $\mathfrak{T}\models \Psi$.
Furthermore,
Even has a winning strategy in $G(\mathfrak{T}\times \Psi)$ from each position
$(s, \phi)$ of $\mathfrak{T}\times \Psi$, where $s$ is a state of $\mathfrak{T}$ and $\phi$ is a subformula of $\Psi$, if and only
if $s \models \phi$.
\end{defn}

\begin{defn} 
The model-checking parity games generated by a formula $\Psi$ are encoded as trees by assigning propositional variables $E_i$ to positions belonging to Even with priority $i$ and $O_i$ to positions belonging to Odd of priority $i$. Encoding more data about the provenance of positions will be useful in the second half of this paper: we will mark modal positions, that is to say positions $s \times \Box \phi$ and $s\times \Diamond \phi$ with a propositional variable $M$.

The winning regions of parity games can be described by $L_\mu$ formulas~\cite{emerson1991tree}, shown to be complete for their index (\emph{i.e.} not equivalent to any formula of lower index)\cite{bradfield1998simplifying}.\\

For an index $I=\{m,...,q\}$ with $m\in\{0,1\}$, define the formula:
$$\mathsf{Parity}_I = \gamma_q X_q. ... \gamma_m X_m. \bigvee_{i\in I} (E_i \wedge \Diamond X_i) \vee (O_i \wedge \Box X_i)$$

where $\gamma_i$ is $\mu$ for odd $i$ and $\nu$ for even $i$.

\end{defn}

\begin{defn}
$\Phi$ \emph{interpret} $\Psi$ if for all trees $\mathfrak{T}$, it is the case that
$\mathfrak{T}\models \Psi$ if and only if $\mathfrak{T}\times \Psi\models \Phi$.
\end{defn}
\begin{thm}\cite{emerson1991tree,walukiewicz2002monadic}
For all $\Psi$ of index $I$, $\mathsf{Parity}_I$ interprets $\Psi$.
\end{thm}

\section{The descriptive complexity of parity games}\label{first-part}

The last section recalled the classic result that
the syntactic complexity of $\Psi$ is an upper bound on the descriptive complexity of its model-checking game.
Here we show that  if a formula of index $I$ interprets $\Psi$,
then the formula $\Psi$ is equivalent to one of index $I$. In other words,
the semantic complexity of a $L_\mu$ formula is a lower bound to the descriptive complexity of its model-checking games.

\begin{thm}\label{converse}
 Let $\Psi$ be a formula of $L_\mu$. If for some formula $\mathsf{Win}$ and all structures $\mathfrak{T}$ we have $\mathfrak{T}\times\Psi\models\mathsf{Win}$ if and only if
$\mathfrak{T}\models\Psi$, that is to say $\mathsf{Win}$ interprets $\Psi$,
then $\Psi$ is equivalent to a formula which has the same alternation depth as $\mathsf{Win}$.
\end{thm}

The proof depends on a product-like operation on formulas which, if $\mathsf{Win}$ interprets $\Psi$,
yields a formula $\Psi\times \mathsf{Win}$, equivalent
to $\Psi$, with the alternation depth of $\mathsf{Win}$. 
The formula $\Psi\times\mathsf{Win}$ is built with the intention that the parity game  arena $(\mathfrak{T}\times\Psi)\times\mathsf{Win}$ is the same as
$\mathfrak{T}\times(\Psi\times\mathsf{Win})$.
The choice of overloading the notation $\times$ is meant to emphasize this associativity. Note however
the type of the objects in these statements: $\Psi \times \mathsf{Win}$ is a formula if $\Psi$ and $\mathsf{Win}$ are both formulas while $\mathfrak{T}\times \Psi$ is a
tree if $\mathfrak{T}$ is a tree.
 

\begin{defn} \emph{($\Psi\times\mathsf{Win}$)}\label{product}
Let $\Psi$ be a formula of $L_\mu$ with index $J$ and priority assignment $\Omega_\Psi$, and let $\mathsf{Win}$ be a formula over propositional variables $\mathcal{P}_{\mathsf{Win}}$ with priority assignment $\Omega_{\mathsf{Win}}$ with co-domain $I$ with minimal element $m$.
Define $\Psi\times\mathsf{Win}$, using a fresh set of fixpoint variables $W_{\phi \times X}$ where $\phi$ ranges over subformulas of $\Psi$ and $X$ ranges over the fixpoint variables of $\mathsf{Win}$, as follows:
\begin{align*}
P \times E_m                  = \neg P &\text { and }
 P \times O_m                 = P \text{ where } P \text { is a literal}; \\
 \phi \times P = \top \text{ if for each state } s &\text{ of any tree, the position }(s,\phi) \text{ satisfies } P \in \mathcal{P}_{\mathsf{Win}}. \\ \text{In particular:} \\
X \times E_{\Omega_\Psi(X)}  &=\phi \wedge \psi \times O_m = \phi \vee \psi \times E_m   = \top; \\
\Diamond \phi \times E_m  &= \Box \phi \times O_m  = \top \times O_m = \bot \times E_m = \top; \\
\phi \times P               &= \bot \text{ for }P\in \mathcal{P}_{\mathsf{Win}} \text{ otherwise.}\\
\Box\phi \times \Box \psi = \Diamond \phi \times \Box \psi  &= \Box(\phi \times \psi);\\
\phi \times \Box \psi       = \bigwedge_{\phi'\in \mathit{im}(\phi)} &\phi'\times \psi \text{ where } \mathit{im}(\phi) \text{ is the set of immediate subformulas of } \phi;\\
\Box\phi \times \Diamond \psi =  \Diamond \phi \times \Diamond \psi &= \Diamond(\phi \times \psi);\\
\phi \times \Diamond \psi   = \bigvee_{\phi'\in\mathit{im}(\phi)} &\phi' \times \psi \text{ where } \mathit{im}(\phi) \text{ is the set of immediate subformulas of } \phi;\\
\text{ If reached } \text{computing a subformula of }&\mu W_{\phi\times X}. (\phi \times \psi_{X})\text{ (or } \nu W_{\phi\times X}. (\phi \times \psi_{X}) \text{ )} \\
\text{but not a subformula } \nu W_{\phi'\times Z}.\phi_{W_{\phi'\times Z}} \text{ (or } & \mu W_{\phi'\times Z}.\phi_{W_{\phi'\times Z}} \text{) thereof with } \Omega(Z)>\Omega(X) \text{ then:}\\
\phi \times X               &= W_{\phi\times X} \text{, else:}\\
                            &= \mu W_{\phi\times X } .(\phi \times \phi_{X}) \text{ if } X \text{ is a }\mu \text{ variable};\\
                            &= \nu W_{\phi\times X}. (\phi \times \phi_{X}) \text{ if } X \text{ is a } \nu \text{ variable.} \\
\phi \times \mu X. \psi_X   &= \mu W_{\phi\times X}. \phi \times \phi_X;\\
\phi \times \nu X. \psi_X   &= \nu W_{\phi\times X}. \phi \times \phi_X;\\
\phi \times \psi \wedge \psi' &= (\phi \times \psi)\wedge (\phi \times \psi');\\
\phi \times \psi \vee \psi' &= (\phi \times \psi)\vee (\phi \times \psi').
\end{align*}
where
$\Omega_{\Psi\times\mathsf{Win}}(W_{X \times \phi})=\Omega_{\mathsf{Win}}(X)$.
Superfluous fixpoint bindings which do not bind any free variables can then be removed.
The construction terminates as computations of $\phi\times X$ will eventually result in the variable $W_{\phi\times X}$ being bound.

\end{defn}

\begin{proof}\emph{(Theorem \ref{converse})}
For the correctness of this construction, it is sufficient to compare rule by rule the parity games on $(\mathfrak{T}\times \Psi)\times \mathsf{Win}$ and $\mathfrak{T}\times (\Psi \times \mathsf{Win})$.

{\bf Case $P\times E_m=\neg P$:} The position $((s,P), E_m)$ is winning for Even exactly when $(s,P)$ belongs to Even: when $s\models \neg P$. The position $(s,\neg P)$ is also winning for Even exactly when $s\models \neg P$.

{\bf Case $P\times O_m = P$:} The position $((s,P),O_m)$ is winning for Even exactly when $(s,P)$ belongs to Odd: when $s\models P$. The position $(s,P)$ is also winning for Even exactly when $s\models P$.

{\bf Case $\phi\times P$:} For $\phi\notin \mathcal{P}_{\mathsf{Win}}$, whether $(s,\phi)$ satisfies $P$ can only depend on $\phi$ and not $s$. If for all states $s$ of any tree, $(s,\phi)$ satisfies $P$, then
$((s,\phi),P)$ is winning for Even, as is $(s,\top)$. Otherwise $((s,\phi),P)$ is winning for Odd, as is $(s,\bot)$.

{\bf Case $\Box \phi \times \Box \psi = \Box (\phi \times \psi)$:}
$((s,\Box \phi),\Box \psi)$ and $((s,\Diamond \phi),\Box \psi)$ are both positions belonging
to Odd with successors $((s', \phi),\psi)$ for $s'$ a successor of $s$. On the other hand $(s,\Box (\phi,\psi))$ has successors $(s',\phi \times \psi)$ for $s'$ a successor of $s$.
The case for $\Box \phi \times \Diamond \phi= \Diamond \phi \times \Diamond \psi=\Diamond (\phi\times \psi)$ is similar.

{\bf Case $\phi \times \Box \psi= \bigwedge_{\phi'\in \mathit{im}(\phi)} \phi' \times \psi$:} The position $((s,\phi),\Box \psi)$ belongs to Odd and has successors $((s,\phi'),\psi)$ where $\phi'$ is an immediate subformula of $\phi$. On the other hand, $(s,  \bigwedge_{\phi'\in \mathit{im}(\phi)} \phi' \times \psi)$ also belongs to Odd and has successors $(s,\phi'\times \psi)$ where $\phi'$ is an immediate subformula of $\phi$.
The case for $\phi \times \Diamond \psi = \bigvee_{\phi'\in \mathit{im}(\phi)} \phi'\times \psi$ is similar.

{\bf Case for fixpoints:} $((s,\phi),X)$ has a unique successor $((s,\phi),\phi_X))$ where $\phi_X$ is the formula binding $X$ and is of the priority of $X$. On the other hand $(s,W_{\phi\times X})$ and $(s,\gamma W_{\phi\times X})$ for $\gamma\in \{\mu,\nu\}$ have a unique successor $(s,\phi\times \phi_X)$, and is of the priority of $W_{\phi \times X}$ which is the same as the priority of $X$.

{\bf Case for $\phi\times \psi \wedge \psi'$:} The position $((s,\phi),\psi \wedge \psi')$ belongs to Odd and has successors $((s,\phi),\psi)$ and $((s,\phi),\psi')$. The position $(s,(\phi \times \psi) \wedge (\phi \times \psi')$ also belongs to Odd and has successors $(s,\phi\times \psi)$ and $(s,\phi \times \psi')$.
The case for $(\phi \times \psi \vee \psi')$ is similar.\\

The formula $\Psi \times \mathsf{Win}$ inherits its priority assignment from $\mathsf{Win}$, and the condition on the introduction of fixpoint variables guarantees that the priority assignment is order-preserving.
Then we have that $\mathfrak{T}\models \Psi \Leftrightarrow \mathfrak{T}\times\Psi\models \mathsf{Win} \Leftrightarrow \mathfrak{T}\models \Psi \times \mathsf{Win}$: the formula
$\Psi\times\mathsf{Win}$, of index $I$, is
equivalent to $\Psi$.
\end{proof}

This concludes the argument that if
$\Psi$ can be interpreted by a formula $\Phi$ with index $I$, then $\Psi$ is semantically of index $I$.
Since $\Psi\times \Phi$ is an effective transformation that turns $\Psi$ into a formula which is syntactically of index $I$,
a formula can be simplified whenever a suitable interpreting formula is found. The following example shows how this can be used to argue that a formula is semantically simple without recourse to the formula's semantics.

\begin{example}\label{ex}
Consider this formula $$\Psi = \mu X. \nu Y. \mu Z. (A \wedge \Diamond Y) \vee (B \wedge \Diamond (Z\wedge \phi)) \vee (C \wedge \Box X)$$ where $\phi \in \Pi^{\mu}_2$ and does not have variables $X, Y, Z$. We use Theorem \ref{converse} to argue that this formula, with seemingly opaque semantics, is semantically $\Pi^{\mu}_2$.

The first thing to note is that a winning strategy $\sigma$ for Even in a model-checking game of $\Psi$ has the following property:
a play that agrees with $\sigma$ eventually reaches a point from which no play that agrees with $\sigma$ sees a position $s\times X$. We can describe the winning regions of such games with the following formula:
$$
 \mathsf{Interpretor} = \mu X. \bigvee_{i\in \{1,2,3\}} (E_i \wedge \Diamond X \vee O_i \wedge \Box X) \vee 
 \nu X_2. \mu X_1. \bigvee_{i\in \{2,1\}} (E_i \wedge \Diamond X_i \vee O_i \wedge \Box X_i)
$$

The intuition here is that the first part of the Interpretor formula describes Even's strategy until she can guarantee that a position of priority $3$ will not be seen again; the second part of the formula is just the usual $\mathsf{Parity}_{2,1}$ formula.
A strategy in the model-checking game for this formula translates directly into the model-checking game of $\mathsf{Parity}_{3,2,1}$. On the other hand, if Even has a winning strategy in a model-checking game $\mathfrak{T}\times \mathsf{Parity}_{3,2,1}$ with the above property, then she also has a winning strategy in $\mathfrak{T}\times \mathsf{Interpretor}$: once a play reaches the point from where her strategy no longer sees priority $3$, she moves to the second part of the formula.

The $\Pi^{\mu}_2$ formula $\mathsf{Interpretor}$ therefore interprets $\Psi$. From Theorem \ref{converse}, $\Psi$ is semantically $\Pi^{\mu}_2$. The mechanics of the transformation can come in the way of clarity, but for the curious reader the tidied-up version of the transformed formula follows. Note how the structure of $\mathsf{Interpretor}$ informs the structure of this new formula.
$$
 \mu X. (A \wedge \Diamond X) \vee (B \wedge \Diamond (X\wedge \phi)) \vee (C \wedge \Box X) \vee \\
  \nu Y. \mu Z. (A \wedge \Diamond Y) \vee (B \wedge \Diamond (Z\wedge \phi)) \vee (C \wedge \Box \bot)
$$
\end{example}

\section{Interpretation theorems}\label{second-part}

 So far, we have seen that if a formula $\Psi$ is interpreted by a formula of index $I$, then $\Psi$ is itself semantically of index $I$. This extension of the $L_\mu$-parity game relationship is natural and not too surprising. 
This section considers the converse: when is it the case that a formula semantically of index $I$ can be interpreted by a formula of syntactic index $I$? The formula in Example \ref{ex} shows that this may happen for  syntactic reasons. Here we discuss whether it holds in general. This conjecture can be studied with respect to any index $I$:

\begin{conjecture}\label{conj-converse}
 If $\Psi$ is semantically of index $I$, then there is a formula $\Phi$ of syntactic index $I$ such that for all structures $\mathfrak{T}$, $\mathfrak{T}\times \Psi \models \Phi$ if and only if
 $\mathfrak{T}\models \Psi$.
\end{conjecture}

If this conjecture is found to be true in general, then the descriptive complexity of the model-checking games of a $L_\mu$ formula $\Psi$ is exactly the complexity of the simplest formula equivalent to $\Psi$. 

Say that the class of formulas of index $I$ admits an \textit{interpretation theorem} if the above conjecture holds for $I$.
This section considers this conjecture for $\mathit{ML},\Pi^{\mu}_1$ and $\Sigma^{\mu}_2$. It shows that these classes admit interpretation theorems:
\begin{itemize}
 \item $\mathit{ML}$ has an interpretation theorem;
 \item $\Pi^{\mu}_1$ has an interpretation theorem for disjunctive $L_\mu$;
 \item $\Sigma^{\mu}_2$ has an interpretation theorem for disjunctive $L_\mu$;
 \end{itemize}
 
It then considers what can be said in the general case, and shows that
disjunctive $L_\mu$ alternation classes all have an interpretation theorem for input in co-disjuntive form.

In each case, the proof strategy relies on defining a characteristic game for each index $I$, with descriptive complexity $I$, and characterising formulas with semantic index $I$ as those whose model-checking parity games corresponds to the characteristic game.

\subsection{Interpretation theorem for $\mathit{ML}$}\label{sec-interpret-ml}

This section shows that conjecture \ref{conj-converse} holds for semantically modal formulas:

\begin{thm}\label{thm-eit-ml}\emph{(Interpretation Theorem for $\mathit{ML}$)}

 Let $\Psi$ be a semantically modal formula, then there exists (effectively) a modal formula $\Phi$ such that $\Phi$ interprets $\Psi$.
\end{thm}

The proof relies on the idea that modal formulas are local: a modal formula holds in a tree if and only if it holds in its initial fragment of depth $m$, the modal depth of the formula given by the maximal nesting depth of modalities. This is also true for formulas equivalent to modal formulas of modal depth $m$. The characteristic game for semantically modal formulas simulates truncating the tree at depth $m$.

\begin{lemma}
If $\Psi\in L_{\mu}$ is equivalent to a modal formula $\Phi$ of modal depth $m$, then for all $\mathcal{T}$, it is the case that $\mathcal{T}\models \Psi$ if and only if $\mathcal{T}|^m \models \Psi$ where $\mathcal{T}|^m$ is the initial fragment of $\mathcal{T}$, up to depth $m$.
\end{lemma}

\begin{proof}
$\Phi$ only has $m$ nested modalities, so the model-checking game of $\Phi$ can not reach positions further than $m$ steps from the root. If two trees agree up to depth $m$, they therefore agree on $\Phi$. Since $\Psi$ is equivalent to $\Phi$, they also have to agree on $\Psi$.
\end{proof}

We now define a variation of a parity game which will allow us to simulate playing the model-checking game on a truncated structure $\mathcal{T}|^n$ rather than on $\mathcal{T}$.

\begin{defn}
A $n$-bounded parity game is a parity game, augmented with a counter, of which some positions $M\in V$ are marked. $M$ must be such that every infinite path goes through $M$. Then, a play of this game is a play of the parity game, except that the counter, initially at $n$, is decremented whenever a position in $M$ is seen. If the play reaches $M$ at counter value $0$, then the owner of the position loses the game.

If the longest path between two positions in $M$ is no longer than $p$, the winning regions of these games are described by the modal formula $\mathsf{Bounded}_{p,m}$ defined inductively as follows:
\begin{align*}
\mathsf{Bounded}_{0,b} = &\bot \\
\mathsf{Bounded}_{a,0} = &(E_i \wedge \neg M \wedge \Diamond \mathsf{Bounded}_{a-1, 0}) \vee  \\
 & (E_i \wedge  M \wedge \bot) \vee \\
 & (O_i \wedge \neg M \wedge \Box \mathsf{Bounded}_{a-1,0}) \vee \\
 & (O_i \wedge M \wedge \top) \\
\mathsf{Bounded}_{a,b} = & (E_i \wedge \neg M \wedge \Diamond \mathsf{Bounded}_{a-1, b}) \vee  \\
 & (E_i \wedge  M \wedge \Diamond \mathsf{Bounded}_{p,b-1}) \vee \\
& (O_i \wedge \neg M \wedge \Box \mathsf{Bounded}_{a-1,b}) \vee \\
& (O_i \wedge M \wedge \Box \mathsf{Bounded}_{p,b-1})
\end{align*}
\end{defn}

\begin{proof}(of Theorem \ref{thm-eit-ml})
Let $\Psi$ be equivalent to a modal formula of modal depth $m$. Let $p$ be the longest path in the parse-tree of $\Psi$ (where $\phi_X$ is taken to be the child of $X$) without modalities. Since $\Psi$ is taken to be guarded, $p$ is finite.
Then the model-checking games of $\Psi$ are $m$-bounded parity games where $M$ is the set of positions $s \times  \phi$ where $\phi$ is a modal formula: this game simulates the parity game $\mathcal{T}|^m \times \Psi$. To see this, it suffices to observe that the two games are identical until $m$ modalities are seen, after which in the $m$-bounded game positions $\Diamond \phi$ are winning for Odd, simulating reaching a position without successors at Even's turn, while $\Box \phi$ is winning for Even, simulating reaching a position without successors at Odd's turn.
 Then, noting that the longest path between two positions in $M$ in any model-checking parity game of $\Psi$ is of length at most $p$, $\mathsf{Bounded}_{p,m}$ interprets $\Psi$.
\end{proof}

Hence the modal fragment of $L_\mu$ has an effective interpretation theorem: any semantically modal formula can be interpreted by a syntactically modal formula. This means that all semantically modal formulas, no matter how high their syntactic index, generate a class of model-checking games
with modal descriptive complexity. The interpreting formula depends on the size of $\Psi$ as well as its semantic modal depth, but is computable from the formula ~\cite{lehtinen2015deciding}. 

\subsection{Interpretation theorem for $\Pi^{\mu}_1$}\label{sec-interpret-pi}

The conjecture \ref {conj-converse} also holds for \emph{disjunctive formulas} that are semantically $\Pi^{\mu}_1$. 

\begin{thm}\label{thm-eit-pi}\emph{(Interpretation theorem for $\Pi^{\mu}_1$)}

 Let $\Psi$ be a disjunctive formula which is semantically in $\Pi^{\mu}_1$. Then there exists (effectively) a $\Pi^{\mu}_1$ formula $\Phi$ such that $\Phi$ interprets $\Psi$.
\end{thm}

The decision procedure for $\Pi^{\mu}_1$ in~\cite{lehtinen2015deciding} shows that in disjunctive semantically $\Pi^{\mu}_1$
formulas, every satisfiable $\mu$-subformula can be replaced with the same subformula bound by $\nu$, while unsatisfiable ones can be replaced by $\bot$. This can be translated into
the formula describing the winning regions of the model-checking games of the formula, simply by substituting the least fixpoint in the parity game formula
corresponding to a subformula $\mu X.\phi$ with either $\bot$
or a greatest fixpoint, depending on whether the subformula is unsatisfiable.

For example, the appropriate formula describing the winning regions of a disjunctive formula in which all $\mu$-bound subformulas
are satisfiable would simply be:
$$\nu Y. (E_e \wedge \Diamond Y) \vee (E_o \wedge \Diamond Y) \vee (O_e \wedge \Box Y) \vee (O_o \wedge \Box Y) $$

Where $E_e$ ($E_o$) is the disjunction of $E_i$ for even (odd) $i$; $O_e$ ($O_o$) is the disjunction of $O_i$ for even (odd) $i$.
For $\mu$-bound subformulas which are unsatisfiable, it is enough to turn the clause $E_{i} \wedge \Diamond X_i$ corresponding
to the $\mu$-variable $X_i$ in question into $E_{i}\wedge \bot$.

\begin{remark}
 Note that the encoding of the model-checking game arena can incorporate data about the provenance of a node, rather than just its priority, for instance by using additional propositional variables $E_X$ to indicate positions which stem from the fixpoint variable $X$. To describe the winning regions of games encoded in this manner, we use $E_X \wedge \Diamond Y_{\Omega(X)}$ instead of $E_i\wedge \Diamond Y_i$.
\end{remark}

Recall that for modal formulas, the interpreting formula depends only on the modal rank of the formula and the length of the longest path  without fixpoints in the formula.
Since both of these can be bound in relation to the size of the formula, all semantically modal $L_\mu$ formulas of the same size can be interpreted by the same
formula. This cannot be said for semantically $\Pi^{\mu}_1$ formulas: the interpreting formulas depend on which fixpoint subformulas are unsatisfiable
and which are interchangeable with $\nu$. It remains an open question whether a \emph{uniform} interpreting formula could be devised for all semantically $\Pi^{\mu}_1$ formulas in
disjunctive form. It is also open whether the restriction to disjunctive formulas can be lifted.

\subsection{Interpretation theorem for $\Sigma^{\mu}_2$}\label{sec-interpret-alt-free}

The previous sections establish effective interpretation theorems for $\Pi^{\mu}_1$, restricted to disjunctive formulas, and for the modal fragment of $L_\mu$.
This section shifts the focus onto an alternation class not known to be decidable: $\Sigma^{\mu}_2$.

\begin{thm}\emph{(Interpretation Theorem for $\Sigma^{\mu}_2$)}\label{int-thm-sigma2}

 If a disjunctive formula $\Psi$ is semantically in $\Sigma^{\mu}_2$, then it is interpreted by a $\Sigma^{\mu}_2$ formula.
\end{thm}

The proof argues that first part of the decidability proof of $\Sigma^{\mu}_2$ for $\Pi^{\mu}_2$ presented in~\cite{lehtinen2017accepted} translates into an interpretation theorem for $\Sigma^{\mu}_2$.
 Similarly to the modal case, a variation on parity games yields the model-checking games for semantically $\Sigma^{\mu}_2$ formulas. Their winning regions are described by a $\Sigma^{\mu}_2$ formula.
 Here we briefly recall the $n$-challenge game, and the result that its winning regions are described by a $\Sigma^{\mu}_2$ formula.

\begin{defn}
Fix a formula $\Psi$ in disjunctive form, of index $\{q,...,0\}$. Let $I = \{q,...,0\}$ if $q$ is even and $\{q+1,q,...,0\}$ otherwise.
Write $I_e$ for the even priorities in $I$.
 
The $n$-challenge game consists of a normal parity game augmented with a set of challenges, one for each even priority $i$. A challenge can
either be \emph{open} or \emph{met} and has a counter $c_i$ attached to it. 
Each counter is initialised to $n$, and decremented when the corresponding challenge is opened. The Odd player
can at any point open challenges of which the counter is non-zero, but he must do so in decreasing order: an $i$-challenge can only be opened if
every $j$-challenge for $j>i$ is opened. 
When a play encounters the priority $j$ while the $j$-challenge is open, the challenge is said to be met. All $i$-challenges for $i<j$ are then \emph{reset}. This means that
the counters $c_i$ are set back to $n$ and marked \emph{met}.

A play of this game is a play in a parity game, augmented with the challenge and counter configuration at each step. A play with dominant priority $d$ is winning for
Even if either $d$ is even or if every opened $d+1$ challenge is eventually met or reset.
\end{defn}

\begin{proof}\emph{(Theorem \ref{int-thm-sigma2})}
It is known from~\cite{lehtinen2017accepted} that disjunctive $\Psi$ is equivalent to a $\Sigma^{\mu}_2$ formula if an only if there is an $n$ such that for all trees $\mathfrak{T}$, Even wins the $n$-challenge game on $\mathfrak{T}\times \Psi$ if and only if she wins the parity game. The winning regions of the $n$-challenge games on parity game arenas of index $I$ are described by a $\Sigma^{\mu}_2$ formula $\mathsf{Challenge}_{I}^n$. Hence, if $\Psi$ is a disjunctive formula semantically in $\Sigma^{\mu}_2$, then there is an $n$ such that for all $\mathfrak{T}$, $\mathfrak{T}\models \Psi$ if and only if $\mathfrak{T}\times \Psi\models \mathsf{Challenge}_{I}^n$. That is to say, $\mathsf{Challenge}_{I}^n$ interprets $\Psi$.
\end{proof}

This concludes the argument that for formulas up to semantic index $\{1,0\}$, the descriptive complexity of the model-checking games coincides exactly with the semantic complexity of the formula. The dependence of these theorems on
disjunctive form highlights the importance of this normal form for the index problem.

\subsection{General interpretation theorem for disjunctive $L_\mu$}\label{sec-general-interpretation}

This section generalises the argument
used for $\Sigma^{\mu}_2$. With some concessions over the structure of target formulas, that is
to say only considering disjunctive alternation classes, this yields a general interpretation theorem for co-disjunctive formulas.\\

We generalise the $n$-challenge game by considering a set of challenges (one per target priority) per input priority rather than a single challenge.
Fix $J$ to be the \emph{input} index, \emph{i.e.} the index of the input formula, while $I$ is the \emph{target} index, \emph{i.e.} the
index of the target alternation class. Then, for the pair $J, I$, we defines a parameterised challenge game such that an input formula $\Psi$ in co-disjuntive form
of index $J$ is interpreted by the disjunctive formula (of index $I$) describing the winning regions of these games if and only if $\Psi$ is equivalent to a disjunctive formula of index $I$.

As usual, this game is played on a parity game arena with priorities from $J$. Unlike for $\Sigma^{\mu}_2$, the challenging player is now Even, to allow the challenger to use the fact that in disjunctive form she has strategies which see one play per branch.
The input formula on the other hand is in co-disjunctive form, to yield such strategies for Odd.

\begin{defn}\emph{(Generalised challenge games)}
A \emph{challenge configuration} $(\bar a,\bar c)$ consists of $|I_o|\times |J_o|$ challenges and counters, where $I_o$ and $J_o$ are the odd priorities in $I$ and $J$ respectively. 
 Write $a_{i,j} = \mathit{met}$ or $\mathit{open}$
for $i\in I$ and $j\in J$ to indicate whether the $i$-level challenge on $j$ is open. Each challenge $a_{i,j}$ is attached to a positive integer counter value $c_{i,j}$, bounded by $n$.

Given a configuration $(\bar a, \bar c)$, the least $j$ for which $a_{i,j}$ is open (for any $i$) is the \emph{priority} of the challenge configuration,
and the highest level $i$ at which $a_{i,j}$ is open is its \emph{level}.
A valid challenge configuration respects the following constraint:
if $a_{i,j}=\mathit{open}$ then $a_{i,k}=\mathit{open}$ for all $k>j$. That is to say, challenges are opened in decreasing order.

The \emph{game configuration} consists of a position in the parity game and a valid challenge configuration.
The progress of the game can be divided into two rounds: in the first round Even opens or resets challenges, while the second round is a step in the parity game.
In the first round her possible actions are:
\begin{itemize}
 \item To $k$-\emph{reset} for any odd $k\in I_o$, setting $a_{i,j}:=\mathit{met}$ and $c_{i,j}:=n$ for all $i\leq k$ and all $j\in J$;
 \item To open at any level $i$, challenges up to any $p$, as long as the counters allow it: for all $j\geq p$ such that $a_{i,j}=\mathit{met}$, set $a_{i,j}:=\mathit{open}$ and
 $c_{i,j}:=c_{i,j}-1$ if $c_{i,j}>0$.
\end{itemize}
Then, in the second round, the player whose turn it is in the parity game picks a successor position. If the underlying parity game ends in a terminal state, then the
winner of the underlying parity game immediately wins the challenge game, too.
The challenge configuration is updated according to the priority $p$ of this new position:
\begin{itemize}
 \item $a_{i,j}:=\mathit{met}$ for all $j\leq p$, all $i\in I_o$;
 \item $c_{i,j}:= n$ for all $j<p$, all $i\in I_o$.
\end{itemize}
If $c_{i,p}=0$ for some $i$, then the game ends immediately with a win for Odd.

A play is a potentially infinite sequence of game configurations: an underlying parity game play augmented with challenge configurations. The initial challenge-configuration is $a_{i,j}=met$  and $c_{i,j}=n$ for all $i,j$.
The dominant priority of the parity game play is also
the dominant priority of the challenge game play. An infinite play with dominant priority $d$ is winning for Odd if:
 $d$ is odd, or all $a_{i,d+1}$ challenges are in the $\mathit{met}$ state infinitely often. 
\end{defn}

\begin{lemma}\label{general-formula}
 The winning regions of a generalised $n$-challenge game for $J,I$ are described by a disjunctive $L_\mu$ formula with index $I$.
\end{lemma}

The construction of the formula $\mathsf{GChallenge}^n_{J,I}$ is similar in spirit to the one described in~\cite{lehtinen2017accepted}, although slightly more involved, given the additional input priorities to account for.

\begin{lemma}\label{general}
A co-disjunctive $L_\mu$ formula $\Psi$ of index $J$ is interpreted by the $I$-index formula $\mathsf{GChallenge}^n_{J,I}$ if it is equivalent to a disjunctive formula of index $I$.
\end{lemma}

Again, the proof follows the same structure as the proof of the $\Sigma^{\mu}_2$ case. In brief, it assumes that $\Psi$ of index $J$ is equivalent to some disjunctive $\Phi$ of index $I$ and that for all $m$, there is a structure $\mathfrak{T}$,
such that Even wins the parity game $\mathfrak{T}\times \Psi$
but Odd wins the generalised $m$-challenge game on $\mathfrak{T}\times \Psi$.  It then takes a sufficiently large $m$, and uses a winning strategy $\sigma$ for Even in $\mathfrak{T}\times \Phi$ to define a challenging strategy $\gamma$
for her in the generalised $m$-challenge game on $\mathfrak{T}\times \Psi$. In the $\Sigma^{\mu}_2$ case the challenges were issues when the higher priority had been seen on all plays since the last challenge. This time, different even priorities cause Even to issue challenges at different levels while different odd priorities cause Even to reset at different levels. Then Odd's winning strategy $\tau$ in the $m$-challenge game on $\mathfrak{T}\times \Psi$ is used to add back edges to $\mathfrak{T}$, turning it into a new
structure $\mathfrak{T}'$ which preserves Even's winning strategy $\sigma$
in $\mathfrak{T}'\times \Phi$ while turning $\tau_\gamma$ into a winning strategy in $\mathfrak{T}'\times \Psi$. This contradicts the equivalence of $\Phi$ and $\Psi$. The main technical difference is that we no longer invoke K\"onig's lemma, and instead uses the fact that target formulas are in disjunctive form to find a suitable challenging-strategy for Even.\\

This gives us a general interpretation theorem for all disjunctive $\Sigma^{\mu}$ alternation classes, for co-disjunctive input formulas: the complexity of an equivalent disjunctive formula is an upper bound on the descriptive complexity of the model-checking games generated by the co-disjunctive form of $\Psi$.

\subsubsection*{Comparison with automata-theoretic results}

The readers familiar with the automata-theoretic efforts to tackle the Rabin--Mostowski index problem may find some parallels between this last interpretation theorem and the reduction of the decidability of the  Rabin--Mostowski hierarchy of non-deterministic automata to a boundedness question~\cite{colcombet2008non}. 
In both cases, the core construction gives a way of mapping input priorities to the target priorities. In~\cite{colcombet2008non}, it is a question of guessing mappings between input and target priorities a bounded number of times. Here it is achieved through the parameterised challenges. The proof in~\cite{colcombet2008non} relies on \emph{guidable} automata -- a concept for which the $L_\mu$-theoretic conterpart is conspicuously absent.

I believe this comparison to be a cause for cautious optimism when looking for further interpretation theorems: indeed, the automata-theoretic result is not burdened by the same constraints to disjunctive and co-disjuntive form as Theorem \ref{general}; perhaps these constraints can be lifted in the $L_\mu$ framework and the mapping construction in~\cite{colcombet2008non} could yield a family of interpreting formulas, for which the interpretation theorem does not place restrictions on the syntax of the input formula.

The transferability of the techniques is however not trivial. Disjunctive $L_\mu$ was designed to be the $L_\mu$ counter-part of non-deterministic automata. However, although the model-checking problems for non-deterministic automata and disjunctive $L_\mu$  reduce to each other by encoding unranked trees as ranked trees, this reduction is not necessarily index-preserving. Furthermore, only a small fragment of $L_\mu$, namely disjunctive $L_\mu$ \emph{restricted to existential modalities} seems to allow a construction akin to guidable automata. Although on this restricted fragment the index problem seems similar to the non-deterministic index problem, it seems unlikely that it would suffice to decide the index problem for disjunctive $L_\mu$. Overall, how the disjunctive $L_\mu$ index problem and the non-deterministic parity automata index problem relate to each other is not yet settled.

\section{Discussion}

We have shown that for formulas equivalent to formulas of index $I\in \{\{\},\{0\},\{0,1\}\}$, the descriptive complexity of their model-checking games coincides exactly with $I$. Beyond these alternation classes, the descriptive complexity of the model-checking games of $\Psi$ is an upper bound on the semantic complexity of $\Psi$. Furthermore, for a co-disjunctive formula equivalent to a disjunctive formula of index $I$, the descriptive complexity of its model-checking games is bounded by $I$.

This links the index-problem to the parity game--$L_\mu$ relationship, which is shown to be even more robust than previously accounted for. As a result, the index-problem can be approached by studying the descriptive complexity of model-checking games. On one hand, for practical purposes formulas can be simplified by looking at the model-checking games they generate. On the other hand, interpretation theorems are a stepping stone on the pursuit of further decision procedures.

Besides finding a truly general interpretation theorem which would show that the descriptive complexity of model-checking games coincides exactly with the semantic complexity of formulas, I leave the reader with a couple of related open questions.

 For low alternation classes $\mathit{ML},\Pi^{\mu}_1$ and $\Sigma^{\mu}_2$, some of the existing decision procedures can be recast as interpretation theorems. For both $\Pi^{\mu}_1$ and $\Sigma^{\mu}_2$, the interpretation theorem only applies for disjunctive input formula.
This once again highlights the importance of disjunctive form for the index problem: the descriptive complexity of the model-checking games
of disjunctive formulas appears to be lower than for general $L_\mu$ formulas. Whether this is truly the case, remains an open question.

The generality of the last interpretation theorem comes at a cost.  In
particular, it is not constructive, in the sense that if a co-disjunctive formula is indeed equivalent to a disjunctive formula of index $I$,
even if we find the interpreting disjunctive formula of index $I$, this only allows us to construct a formula of index $I$,
which may not itself be disjunctive.
 I conjecture that a similar theorem could be shown for disjunctive input formulas. As discussed in the previous section, some automata-theoretic techniques~\cite{colcombet2008non} could perhaps be used to address some of the limits of the current techniques.

\bibliographystyle{eptcs}
\bibliography{gandalf}

\end{document}